\documentclass[journal]{IEEEtran}
\newcommand{\fsize}{0.48}
\newcommand{\fsb}{0.43}

\usepackage{algorithmic,algorithm}
\usepackage{amsmath,amssymb,bm}
\usepackage{graphicx, color,epsfig,epstopdf,}
\usepackage{cite}

\newcommand{\qd}{{\bf d}}
\newcommand{\qe}{{\bf e}}
\newcommand{\qf}{{\bf f}}
\newcommand{\qg}{{\bf g}}
\newcommand{\qh}{{\bf h}}

\newcommand{\qt}{{\bf t}}

\newcommand{\qw}{{\bf w}}

\newcommand{\qA}{{\bf A}}

\newcommand{\qI}{{\bf I}}

\newcommand{\qQ}{{\bf Q}}

\newcommand{\qT}{{\bf T}}

\newcommand{\qzero}{{\bf 0}}
\newcommand{\qone}{{\bf 1}}

 \newtheorem{corollary}{Corollary}

\newcommand{\be}{\begin{equation}} \newcommand{\ee}{\end{equation}}
\newcommand{\bea}{\begin{eqnarray}} \newcommand{\eea}{\end{eqnarray}}

\newtheorem{theorem}{Theorem}
\begin{document}

\title{Exploiting Known Interference as Green Signal Power for Downlink Beamforming Optimization}
% \title{Exploiting Known Interference with Data-aided {Downlink} Beamforming Optimization}
% \title{Optimization of Data-aided Linear Precoding with QoS  Constraints}
\author{Christos~Masouros,~\IEEEmembership{Senior Member,~IEEE}, Gan Zheng,~\IEEEmembership{Senior Member,~IEEE} \\
\thanks{Manuscript received May 30, 2014; revised September 19, 2014 and February 9, 2015; accepted April 18, 2015. The
associate editor coordinating the review of this manuscript and approving it for
publication was Prof. Y. W. Hong.}
\thanks{Christos~ Masouros is with Department of Electronic \& Electrical Engineering - University College London, Torrington Place, London WC1E 7JE,
UK, E-mail: {\sf chris.masouros@ieee.org}.}
\thanks{
Gan Zheng is with School of Computer Science and Electronic
Engineering, University of Essex, Colchester,  CO4 3SQ, UK, E-mail:
{\sf ganzheng@essex.ac.uk}. He is also affiliated with
Interdisciplinary Centre for Security, Reliability and Trust (SnT),
  University of Luxembourg, Luxembourg.}
%\date{\today}
\thanks{This work was supported by the Royal Academy of Engineering, UK and the Engineering and Physical Sciences Research Council (EPSRC) project EP/M014150/1.}}

\maketitle
\begin{abstract}
We propose a data-aided transmit beamforming scheme for the multi-user
multiple-input-single-output (MISO) downlink channel.  While
conventional beamforming schemes aim at the minimization of the
transmit power subject to suppressing interference to guarantee
quality of service (QoS) constraints, here we use the knowledge of
both data and channel state information (CSI) at the transmitter to
exploit, rather than suppress, constructive interference. More
specifically, we design a new precoding scheme for the MISO downlink that
minimizes the transmit power for generic phase shift keying (PSK)
modulated signals. The proposed precoder   reduces the transmit
power compared to conventional schemes, by adapting the QoS
constraints to accommodate constructive interference as a source of
useful signal power. By exploiting the power of constructively
interfering symbols, the proposed scheme achieves the required QoS
at lower transmit power. We extend this concept to the signal to
interference plus noise ratio (SINR) balancing problem, where higher
SINR values compared to the conventional SINR balancing optimization
are achieved for given transmit power budgets. In addition, we derive equivalent
virtual multicast formulations for both optimizations, both of which provide insights of the
optimal solution and facilitate the design of a more efficient
solver. Finally,
we propose a robust beamforming technique to deal with imperfect CSI,
that also reduces the transmit power over conventional techniques,
while guaranteeing the required QoS. Our simulation
and analysis show significant power savings for small scale MISO downlink channels with
the proposed data-aided optimization compared to conventional
beamforming optimization.

 \end{abstract}

\section{Introduction}

The power efficiency of wireless transmission links has recently
attracted considerable research interest,  in-line with the global
initiative to reduce the CO$_2$ emissions and operational expenses
of communication systems. Accordingly, transmit beamforming schemes
for power minimization have become particularly relevant for the
downlink channel. Capacity achieving non-linear dirty paper coding
(DPC) techniques \cite{DPC1,DPC2} have been proposed for
pre-subtracting interference prior to transmission. The DPC methods
developed so far are in general complex as they require
sophisticated sphere-search algorithms \cite{SphereDec} to be
employed at the transmitter and assume codewords with infinite
length for the encoding of the data. Their suboptimal counterparts
\cite{THPFisher}-\cite{MIOTHP} offer a complexity reduction at a comparable
performance. Still however, the associated complexity is prohibitive
for their deployment in current communication standards. On the other hand, linear precoding schemes based on channel
inversion \cite{CI}-\cite{tvt} offer the  least complexity, but
their performance is far from achieving the optimum maximum
likelihood bound. Their non-linear adaptation, namely vector
perturbation (VP) precoding \cite{VP}-\cite{VP4} provides a performance
improvement at the expense of an increased complexity since
the search for the optimal perturbation vectors is an NP-hard
problem, typically solved by complex sphere search algorithms at the
transmitter.

More relevant to this work, optimization based techniques that
directly minimize the transmit power subject to quality of service
(QoS) constraints - most commonly to
interference-plus-noise ratio  (SINR) - have been studied for broadcast channels in
\cite{Ott}, where convex optimization problems of such nature were
proposed. Recent works have focused on the utility maximization in MIMO interfering broadcast channels \cite{Luo},\cite{Yu} and full duplex radio \cite{Ng}\color{black}. For the cases of channel state information (CSI) errors, robust versions of these optimizations have been studied in
\cite{Wu} - \cite{GanRobust}. In \cite{Pascual}, a robust
max-min approach was developed for a single-user MIMO system based
on convex optimization. Later in \cite{Wiesel}, the robust
transmission schemes to maximize the compound capacity for single
and multiuser rank-one Ricean MIMO channels were addressed, based on
the uncertainty set in \cite{Wolfowitz}. Robust beamforming for
multiuser multiple-input single-output (MISO) downlink channels with
individual QoS constraints under an imperfect channel covariance
matrix was studied in \cite{Ott}, \cite{Gersh2}. Recently in
\cite{Pascual2}, the optimal power allocation over fixed beamforming
vectors was obtained in the presence of errors in CSI matrices. Most
recently, efficient numerical solutions to find conservative robust
beamforming for multiuser MISO systems with mean-square-error (MSE)
and SINR constraints and different bounded CSI errors have been
developed in \cite{Boche2},\cite{GanRobust}. Moreover, SINR
balancing optimizations have been proposed in \cite{Boche} where the
minimum achievable SINR is maximized, subject to a total transmit
power constraint.

This paper is based on the concept of interference exploitation,
first introduced in  \cite{SCI,CR, ComMag, ComMag2} where analytical
interference classification criteria and low-complexity precoders
based on channel inversion were derived. The analysis showed how
the knowledge of both CSI and data at the base station can be used
to predict and classify interference into constructive and
destructive. The closed-form precoders in \cite{SCI,CR} showed that
by exploiting the constructive part of interference, higher receive
SNRs can be provided without increasing the per user transmit power.
However, when precoders are fully optimized, less is
understood about the performance gain of the constructive
interference approach over the conventional optimization, e.g.,
\cite{Ott}. Accordingly, in this work we aim to improve such optimization-based precoders by exploiting constructive interference as a source of
signal power. By doing so, we further reduce the transmit power
required for guaranteeing the SNR constraints in the optimization,
thus improving the power efficiency of transmission. For clarity, we
list the contributions of this paper as follows:
\begin{enumerate}
\item We introduce a new  linear precoder
design for PSK that a) reduces the transmit power for given
performance compared to existing precoders based on the proposed
constructive interference regions, and b) as opposed to conventional precoders, applies to scenarios with higher numbers of users than transmit antennas,
\item We further adapt this concept to the SINR balancing problem, where higher SINR values compared to the conventional SINR balancing optimization are achieved
for given transmit power budgets,
\item We re-cast the optimization problem as a virtual multicast  optimization problem based on which we derive the structure of the optimal solution and develop an efficient
gradient projection algorithm to solve it,
\item We use the multicast  formulation to derive a robust precoding scheme suitable for imperfect CSI with bounded CSI errors.
\end{enumerate}
We note that, while the following analysis focuses on PSK modulation, the above concept and relevant optimizations can be extended to other modulation formats such as quadrature amplitude modulation (QAM) by adapting the decision thresholds of the constellation to accommodate for constructive interference \cite{ComMag}. It should be stressed however, that the proposed schemes are most useful in high interference scenarios where low order modulation such as BPSK and QPSK are used in the communication standards to ensure reliability \cite{LTE}. In addition, constant envelope modulation such as PSK has received particular interest recently with the emergence of large scale MIMO systems \cite{Marzetta}. All the above motivate our focus on PSK constellations. Finally, we note that, in line with the relevant literature we assume a time division duplex (TDD) transmission here, where the base station directly estimates the downlink channel using uplink pilot symbols and uplink-downlink channel reciprocity.

The rest of the paper is organized as follows. Section II introduces
the system model considered in this paper and briefly describes the
two conventional  optimizations of interest: \textit{power minimization} and
\textit{SINR balancing}. In section III the proposed beamforming optimization
based on constructive interference is detailed for the case of power
minimization, while its multicasting equivalent optimization is
shown in section IV. Section V presents the constructive
interference optimization for the SINR balancing problem. The CSI
robust versions of these optimizations are examined in section VI
for both power minimization and SINR balancing, for the case of bounded
CSI errors. Finally numerical results are illustrated and discussed
in section VII and concluding remarks are given in section
VIII. In the following, we use the terms beamforming and precoding interchangeably, in-line with the most relevant literature.

\section{System Model and Beamforming Optimization}
Consider a $K-$user Gaussian broadcast channel where an $N$-antenna BS
 transmits signals to $K$ single-antenna users. Channel vector, precoding vector, received noise, data and SINR constraints for the $i-$th user are
 denoted as $\qh_i^\dag, \qt_i$,  $n_i$, $d_i= d e^{j\phi_i}$ and $\Gamma_i$ respectively with $d$ denoting the constant amplitude of the PSK modulated symbols and $n_i\sim \mathcal{CN}(0, N_0), \forall
 i$ where $\mathcal{CN}(\mu,\sigma^2)$ denotes the circularly symmetric complex Gaussian distribution with mean $\mu$ and variance $\sigma^2$.
 The received signal at user $i$ is
 \bea\label{rxs}
    y_i &=& {\qh_i^T} \sum_{k=1}^K \qt_k d_k + n_i\nonumber\\
      &=& {\qh_i^T} \sum_{k=1}^K \qt_k  e^{j(\phi_k-\phi_i)} d_i +
      n_i.
 \eea

The receive SINR at the $i$-th receiver for this scenario is given as
  \bea\label{SINR}
    \gamma_i=\frac{|\qh_i^T\qt_i|^2}{\sum_{k=1,k\ne i}|\qh_i^T\qt_k|^2 + N_0}, \forall i.%\nonumber
 \eea

 The transmit signal is
 \be
    \sum_{k=1}^K \qt_k d_k = \sum_{k=1}^K \qt_k e^{j(\phi_k-\phi_i)}
    d_i,
 \ee
 where any of the users' symbols $d_i= d e^{j\phi_i}$ can be taken as reference. Without loss of generality, let us use $d_1= d e^{j\phi_1}$ as reference hereof. Accordingly, the instantaneous transmit power for a signal with envelope $d=1$ is defined as
 \be
    P_T =  \left\| \sum_{k=1}^K \qt_k e^{j(\phi_k-\phi_1)}\right\|^2.
 \ee

\subsection {Power Minimization}
The conventional power minimization precoder, treating all interference as harmful, aims to minimize the average transmit power, subject to interference constraints as shown below \cite{Ott}

  \bea\label{eqn:BC}
    \min_{\{\qt_i\}} && \sum_{i=1}^K\|\qt_i\|^2 \\
    \mbox{s.t.} && \frac{|\qh_i^T\qt_i|^2}{\sum_{k=1,k\ne i}|\qh_i^T\qt_k|^2 + N_0}\ge \Gamma_i, \forall i.\nonumber
 \eea

While optimal from a stochastic viewpoint, the above precoder ignores the fact that, instantaneously, interference can contribute constructively to the received signal power \cite{CR}, and therefore from an instantaneous point of view we later show that it is suboptimal.
\begin{figure*}
    \centering
        \includegraphics[width=0.95\textwidth]{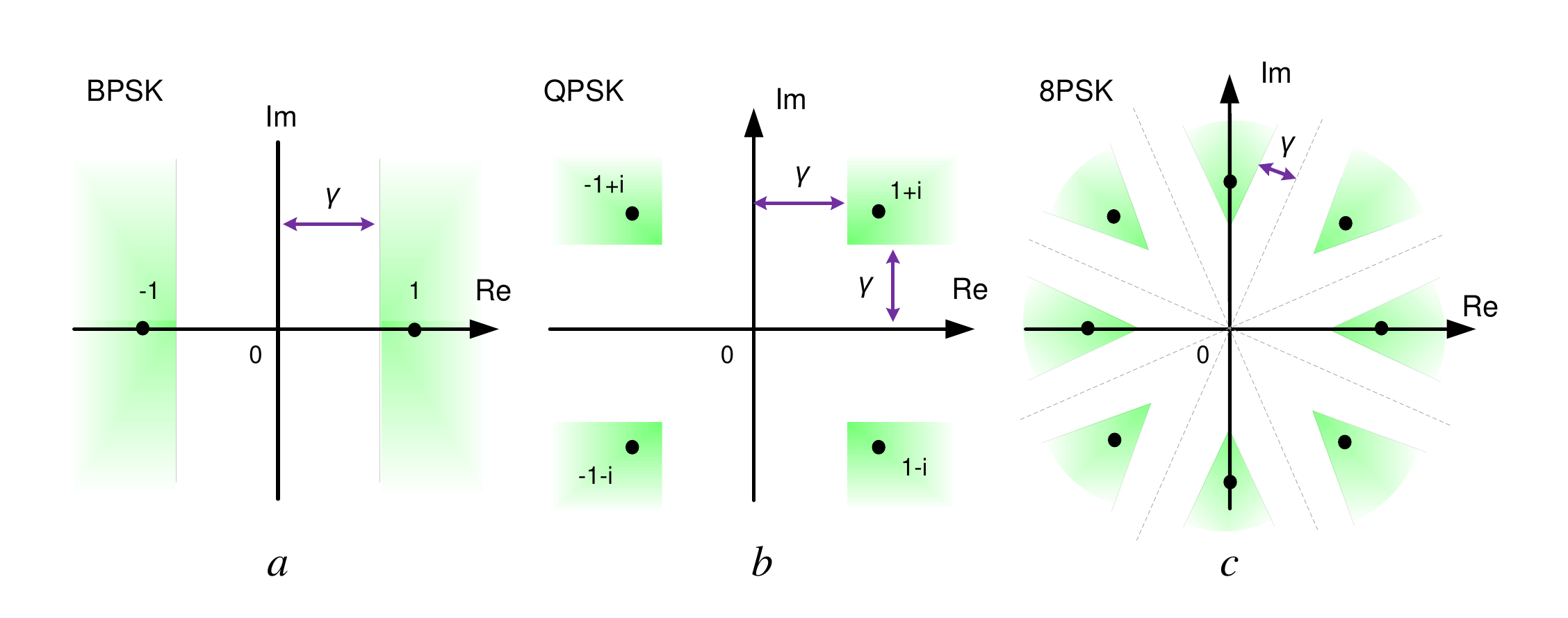}
    \caption{{Constructive interference in a) BPSK, b) QPSK and c) 8PSK}}
    \label{Crit}
\end{figure*}
\subsection{SINR Balancing}
As regards the second optimization that is of interest in this work, SINR balancing maximizes the minimum achievable SINR subject to a transmit power budget, in the form

 \bea\label{eqn:SB2}
    \max_{\qt_i} && \Gamma_t \nonumber\\
    \mbox{s.t.} && \frac{\|\qh_i^T\qt_i\|^2}{\sum_{k=1,k\ne i}\|\qh_i^T\qt_k\|^2 + N_0}\ge \Gamma_t, \forall i.\nonumber\\
    &&\sum_{i=1}^K\|\qt_i\|^2 \leq P
 \eea
where $P$ denotes the total transmit power budget. We note that the above optimization is non-convex and the solution involves more complex iterative solutions \cite{Boche}.

\section{Proposed Optimization based on Constructive Interference and Phase Constraints}

 With the knowledge of the downlink channel and all user's data readily available at the transmitter, and with the aim of exploiting the instantaneous interference, the interference for PSK modulation can be classified to constructive and destructive based on known criteria \cite{CR, ComMag}. For clarity, these are summarized schematically in Fig. \ref{Crit} for the basic PSK constellations. Here the scalar $\gamma$ denotes the threshold distance to the decision variables of the constellation, that relates to the SNR constraint, as detailed in the following. In brief, constructive interference is defined as the interference that moves the received symbols away from the decision thresholds of the constellation. This represents the green areas in the constellations of
 Fig. \ref{Crit} where these are taken with reference to a minimum distance from the decision thresholds as per the SNR constraints. Note that these need not center on the nominal constellation points (the black dots in the figure) for generic SNR constraints. We refer the reader to \cite{SCI, CR, ComMag} for further details on this topic.

 As per the above classification and discussion, the optimizations in \eqref{eqn:BC},\eqref{eqn:SB2} can be modified to take the constructive interference into account in the power minimization. This can be done by imposing interference constraints, not in terms of suppressing the stochastic interference, but rather optimizing instantaneous interference to contribute to the received signal power, thus reducing the required transmit power accordingly. For the case when interference has been aligned by means of precoding vectors ${\qt}_k$ to overlap constructively with the signal of interest, all interference in \eqref{rxs} contributes constructively to the useful signal. Therefore all interference terms form components of the useful signal energy, which is given by the squared magnitude of the full sum in \eqref{rxs}. Accordingly, it has been shown in \cite{CR} that the receive SNR is given as 
 \be\label{CISNR}
\gamma_i= \frac{\left| {\qh_i^T} \sum_{k=1}^K {\qt}_k  d_k\right|^2}{N_0}
 \ee

 \begin{figure*}
    \centering
        \includegraphics[width=\textwidth]{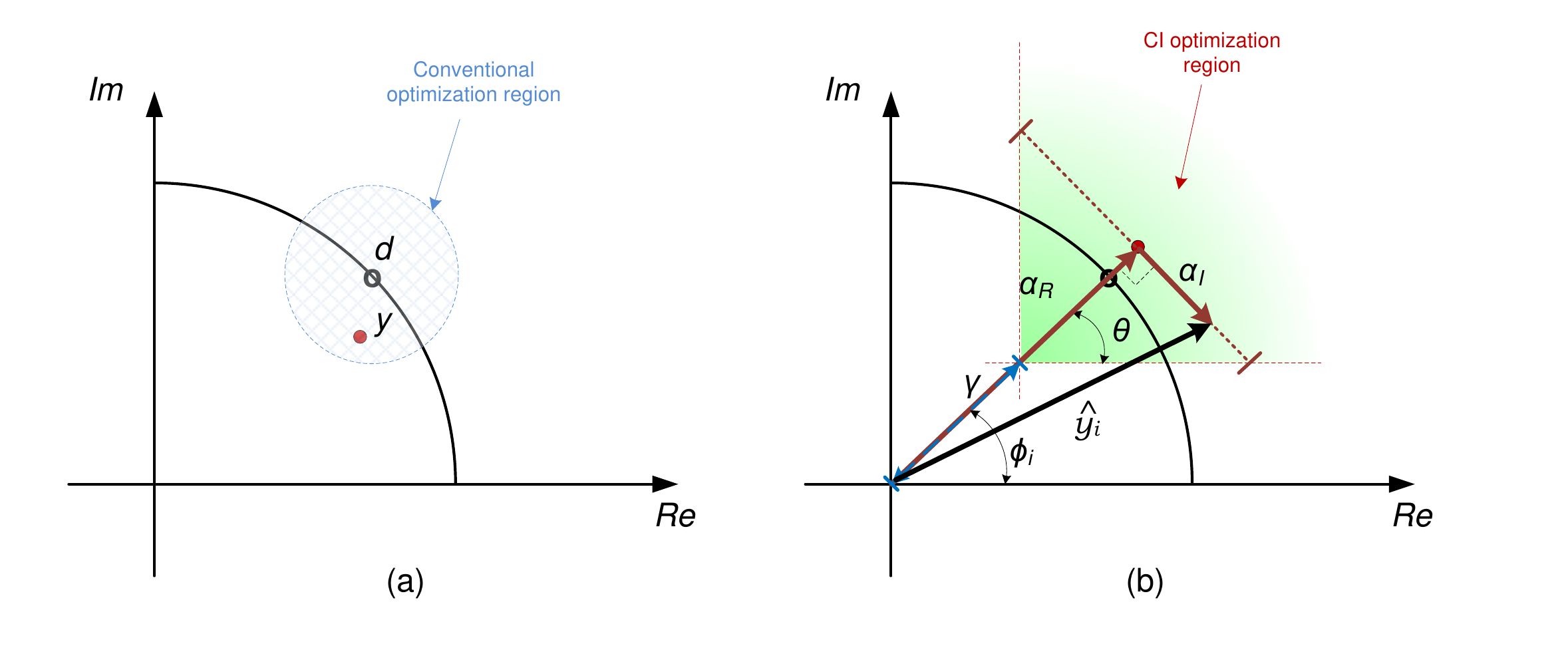}
    \caption{{Optimization regions for (a) conventional precoding and (b) precoding for interference exploitation and generic optimization constraints (QPSK example)}.}
    \label{Constr}
\end{figure*}

where all interference contributes in the useful received signal power. Accordingly, and based on the classification criteria detailed in \cite{SCI} and Fig. \ref{Crit} for $M$-PSK based on constructive interference, where the received interference is aligned to the symbols of interest, the problem can be reformulated as
  \bea\label{CIBCstrict}
    \min_{\{\qt_i\}} &&  \left\| \sum_{k=1}^K \qt_k e^{j(\phi_k-\phi_1)}\right\|^2 \nonumber\\
    \mbox{s.t.}
    &&\angle \left({\qh_i^T} \sum_{k=1}^K \qt_k  d_k\right)=\angle(d_i), \forall i\nonumber\\
     &&\mbox{Re} \left({\qh_i^T} \sum_{k=1}^K \qt_k  e^{j(\phi_k-\phi_i)}\right)\ge
     \sqrt{\Gamma_iN_0}, \forall i,
\eea
where $\mbox{Re}(x)$, $\mbox{Im}(x)$, $sign(x)$ and $\angle x$ denote the real and imaginary part, the sign and angle \color{black} of $x$ respectively. Here the angle of interference is strictly constrained to equal the angle of the symbol of interest. The problem can be equivalently formulated as
   \bea\label{CIBCstrict2}
    \min_{\{\qt_i\}} &&  \left\| \sum_{k=1}^K \qt_k e^{j(\phi_k-\phi_1)}\right\|^2 \nonumber\\
    \mbox{s.t.}
         &&\mbox{Im} \left({\qh_i^T} \sum_{k=1}^K \qt_k  e^{j(\phi_k-\phi_i)}\right)=0, \forall i\nonumber\\
    &&\mbox{Re} \left({\qh_i^T} \sum_{k=1}^K \qt_k  e^{j(\phi_k-\phi_i)}\right)\ge
     \sqrt{\Gamma_iN_0}, \forall i.
 \eea

We note the use of the sum of phase shifted (by the phase of the symbol of interest $\phi_i$) interfering symbols plus the symbol of interest in the above expressions. This serves to isolate the received amplitude and phase shift in the symbol of interest due to interference. Here, the first constraint requires that the interference perfectly
aligns with the phase of the symbol of interest, to ensure that it
overlaps constructively to the useful symbol \cite{CR, ComMag}. The
second constraint requires that this constructive interference is
enough to satisfy the receive SNR threshold. Note that the above two
conditions contain $K$ equations and $K$ inequalities, while there
are $2N$ real variables. Therefore, for the case $N \geq K$ there are always sufficient degrees of
freedom to satisfy these two sets of constraints, while our results in the following show that the proposed can be feasible with high probability even for cases with $N < K$. \color{black} Still, it can be seen that
due to the strict angle constraint, the formulation
\eqref{CIBCstrict2} is more constrained than the constructive interference regions in Fig. \ref{Crit} where the strict angle constraints do not exist. To
obtain a more relaxed optimization for $M-$PSK, let us observe the
constellation example shown in the diagram of Fig. \ref{Constr} for
QPSK. Here, we have used the definitions $\alpha_R=\mbox{Re}
\left({\qh_i^T} \sum_{k=1}^K \qt_k  e^{j(\phi_k-\phi_i)}\right)$ and
$\alpha_I=\mbox{Im} \left({\qh_i^T} \sum_{k=1}^K \qt_k
e^{j(\phi_k-\phi_i)}\right)$ which are the real and imaginary components of $\hat{y}_i\triangleq{\qh_i^T} \sum_{k=1}^K \qt_k  e^{j(\phi_k-\phi_i)}$ in the figure, the received symbol ignoring noise, phase shifted by the phase of the desired symbol. We also define $\tilde{\gamma}=\sqrt{\Gamma_iN_0}$. By means of their definition, $\alpha_R$ and $\alpha_I$ essentially shift the observation of received symbol onto the axis from the origin of the constellation diagram to the constellation symbol of interest. Clearly, $\alpha_R$ provides a measure of the amplification of the received constellation point along the axis of the nominal constellation point due to constructive interference and $\alpha_I$
provides a measure of the angle shift from the original constellation point, i.e. the deviation from the axis of the nominal constellation point with phase $\phi_i$. 

At this point, we should emphasize the key idea in the proposed optimization which is the central strength of the proposed scheme that relaxes the optimization and allows additional reduction in the transmit power. In conventional optimizations, the signal power is optimized subject to SINR constraints. This is equivalent to constraining the interference each user experiences, so that the received symbol is within a certain distance from the nominal constellation symbol. From the view point of multiple users this essentially constrains the transmit vectors such that the received symbol $y$ is contained within a circle (or a hyper-sphere for multidimensional optimizations) around the nominal constellation point $d$, so that the interference caused by the other symbols is limited. This is denoted by the dashed circle around the constellation point in Fig. \ref{Constr}(a). In contrast to this, here by use of the concept of constructive interference we allow a relaxation of $\alpha_R$ and $\alpha_I$ for all transmit symbols, under the condition that the interference caused is constructive, as secured by the constraints of the optimization. This gives rise to the constructive interference sector denoted by the green shaded sector in Fig. \ref{Constr}(b) \cite{CR},\cite{ComMag}. It can be seen that $\alpha_R$ and $\alpha_I$ are allowed to grow infinitely, as long as their ratio is such that the received symbol is contained within the constructive area of the constellation, i.e. the distances from the decision thresholds, as set by the SNR constraints $\Gamma_i$, are not violated. It is clear that the region in Fig. \ref{Constr}(b) is only constrained along the vicinity of the decision thresholds, it therefore extends infinitely to the directions away from the decision thresholds and hence provides a more relaxed optimisation with respect to the conventional region of Fig. \ref{Constr}(a). \color{black}

\begin{table*}\normalsize
    \bea\label{CIBCrelc}
    \min_{\{\qt_i\}} &&  \left\| \sum_{k=1}^K \qt_k e^{j(\phi_k-\phi_1)}\right\|^2  \nonumber\\
    \mbox{s.t.}
    &&\mbox{Re} \left({\qh_i^T} \sum_{k=1}^K \qt_k  e^{j(\phi_k-\phi_i)}\right)\ge
     \sqrt{\Gamma_iN_0}, \forall i\nonumber\\
%          &&-\left(\mbox{Re} \left({\qh_i^T} \sum_{k=1}^K \qt_k  e^{j(\phi_k-\phi_i)}\right)- \sqrt{\Gamma_iN_0}\right)\tan\theta\le\mbox{Im} \left({\qh_i^T} \sum_{k=1}^K \qt_k  d_k/d_i\right)\nonumber\\
&&\left|\mbox{Im} \left({\qh_i^T} \sum_{k=1}^K \qt_k  e^{j(\phi_k-\phi_i)}\right)\right|\le\left(\mbox{Re} \left({\qh_i^T} \sum_{k=1}^K \qt_k  e^{j(\phi_k-\phi_i)}\right)- \sqrt{\Gamma_iN_0}\right)\tan\theta, \forall i \eea
%\hline
\bea\label{CIBCrelc3}
    \min_{\{\qt_i\}} &&  \left\| \sum_{k=1}^K \qt_k e^{j(\phi_k-\phi_1)}\right\|^2  \nonumber\\
    \mbox{s.t.}
&&\left|\mbox{Im} \left({\qh_i^T} \sum_{k=1}^K \qt_k  e^{j(\phi_k-\phi_i)}\right)\right|\le\left(\mbox{Re} \left({\qh_i^T} \sum_{k=1}^K \qt_k  e^{j(\phi_k-\phi_i)}\right)- \sqrt{\Gamma_iN_0}\right)\tan\theta, \forall i \eea

\end{table*}

As regards the constructive area in the constellation, with respect to \eqref{CIBCstrict2}, it can be seen that the angle of
interference need not strictly align with the angle of the useful
signal, as long as it falls within the constructive area of the
constellation. For a given modulation order $M$ the maximum angle
shift in the constructive interference area is given by
$\theta=\pm \pi/M$. Accordingly, to relax the optimization, $\alpha_I$
is allowed to be non-zero as long as the resulting symbol lies
within the constructive area of the constellation. Using basic
geometry we arrive at the optimization problem expressed as in eq. \eqref{CIBCrelc}.

By noting that the last constraint actually incorporates the one for
the real part of the received symbols, the optimization can be
further reduced to the compact form of eq. \eqref{CIBCrelc3}.

 Clearly, the above optimization is equivalent to
  \bea\label{CIBCrel2}
    \min_{\{\qt_i\}} &&  \left\| \sum_{k=1}^K \qt_k e^{j(\phi_k-\phi_1)}\right\|^2 \nonumber\\
    \mbox{s.t.}
    &&\left|\angle \left({\qh_i^T} \sum_{k=1}^K \qt_k  d_k\right)-\angle(d_i)\right|\leq\vartheta , \forall i\nonumber\\
     &&\mbox{Re} \left({\qh_i^T} \sum_{k=1}^K \qt_k  e^{j(\phi_k-\phi_i)}\right)\ge
     \sqrt{\Gamma_iN_0}, \forall i.
\eea
 where for $\vartheta$ we have
  \bea\label{tanphi}
\tan\vartheta=\tan\theta\left(1-\frac{\sqrt{\Gamma_iN_0}}{\mbox{Re} \left({\qh_i^T} \sum_{k=1}^K \qt_k  e^{j(\phi_k-\phi_i)}\right)}\right)
\eea
It can be seen that the above optimization in \eqref{CIBCrelc3} is more relaxed than the zero-angle-shift optimization \eqref{CIBCstrict}, which results in a smaller minimum in the transmit power. Moreover, it contains a number of $K$ inequalities which result in an increased feasibility region compared to the conventional optimization, as will be shown in the following. Problem (\ref{CIBCrelc3}) is a standard
second-order cone program (SOCP), thus can be optimally
solved using numerical software, such as Semudi. However, in the
following section, we derive a more computationally efficient algorithm to solve it. For the illustrative example of BPSK used in the following results, the optimization can be modified to

 \bea\label{CIBC2}
    \min_{\{\qt_i\}} && \left\| \sum_{k=1}^K \qt_k e^{j(\phi_k-\phi_1)}\right\|^2  \\
    \mbox{s.t.}
    &&\mbox{Re} \left({\qh_i^T} \sum_{k=1}^K \qt_k  d_k\right)sign(d_i)\ge
     \sqrt{\Gamma_iN_0}, \forall i,\nonumber
 \eea

 In \eqref{CIBC2} the constraint requires that the interference on each user's symbol has the same sign as the symbol of interest (and is therefore constructive) and that this constructive interference has enough power to satisfy the SNR threshold. For the case of QPSK the same principle needs to be applied separately to the real and imaginary part of the received symbol (see Fig. \ref{Crit}(b)) and hence we have
  \bea\label{CIBC3}
    \min_{\{\qt_i\}} &&  \left\| \sum_{k=1}^K \qt_k e^{j(\phi_k-\phi_1)}\right\|^2 \\
    \mbox{s.t.}
    &&\mbox{Re} \left({\qh_i^T} \sum_{k=1}^K \qt_k  d_k\right)sign(\mbox{Re}(d_i))\ge
     \sqrt{\frac{\Gamma_iN_0}{2}}, \forall i.\nonumber\\
    &&\mbox{Im} \left({\qh_i^T} \sum_{k=1}^K \qt_k  d_k\right)sign(\mbox{Im}(d_i))\ge
     \sqrt{\frac{\Gamma_iN_0}{2}} \forall i.\nonumber
 \eea

\section{A Virtual Multicast Formulation and a New Efficient Algorithm }

\subsection{A Virtual Multicast Formulation of \eqref{CIBCrelc3}}
 The optimization in \eqref{CIBCrelc3} can be re-cast as a virtual multicasting problem \cite{Sidir} according to the following theorem.

\begin{theorem}\label{T1}
 The  broadcast problem (\ref{CIBCrelc3}) with constructive interference is equivalent to the multicast problem
 below
 \bea\label{eqn:MC}
    \min_{\qw} && \|\qw\|^2  \\ \nonumber
    \mbox{s.t.} && \left|\mbox{Im} \left({\tilde\qh_i^T}   \qw \right)\right| \le \left(\mbox{Re} \left({\tilde\qh_i^T} \qw\right)   - \sqrt{\Gamma_iN_0} \right)\tan\theta, \forall i,
    \label{eqn:constraint}
 \eea
where the modified channel is defined as $\tilde \qh_i \triangleq
\qh_i e^{j(\phi_1-\phi_i)}$. To be more specific, the optimal
solutions to (\ref{CIBCrelc3}) and (\ref{eqn:MC}), $\{\qt^*_i\}$
and $\qw^*$ respectively, have the following
relation
 \bea
  \qt_1^* &= &\frac{\qw^*}{K}, \\
  \qt_k &=& \qt_1^*e^{j(\phi_1-\phi_k)} =\frac{\qw^*e^{j(\phi_1-\phi_k)} }{K}, k = 2,
  \cdots, K.
\eea
\end{theorem}
\begin{proof} First we re-write the constraint in \eqref{CIBCrelc3} as \eqref{CIBCrelc4}.
\begin{table*}\normalsize
 \be\label{CIBCrelc4}\left|\mbox{Im} \left({\qh_i^T} e^{j(\phi_1-\phi_i)} \sum_{k=1}^K \qt_k
e^{j(\phi_k-\phi_1)}\right)\right|\le\left(\mbox{Re}
\left({\qh_i^T}e^{j(\phi_1-\phi_i)} \sum_{k=1}^K \qt_k
e^{j(\phi_k-\phi_1)}\right)- \sqrt{\Gamma_iN_0}\right)\tan\theta,
\forall i.\ee
\end{table*}

 Observe that with \eqref{CIBCrelc4},  the composite precoding term $\sum_{k=1}^K \qt_k
 e^{j(\phi_k-\phi_1)}$ in \eqref{CIBCrelc3} can be treated as a single vector $\qw$ precoder, i.e., $\sum_{k=1}^K \qt_k
 e^{j(\phi_k-\phi_1)}=  \qw$. Therefore the multicast reformulation
 \eqref{eqn:MC} follows immediately.

Suppose  the optimal solutions of user 1's precoding vector in
 (\ref{CIBCrelc3}) is $\qt_1^*$. {Without loss of optimality, the other users' precoding vectors
 can be simple rotated versions of $\qt_1^*$, i.e., $\qt_i^*=\qt_1 e^{j(\phi_1-\phi_i)}, i = 2,
  \cdots, K.$ As a special case, we  have  $\frac{\qw^*}{K} =
  \qt_1^*$. This completes the proof.}
%  As a result, the received
%  signal can be expressed as
%  \bea
%    y_i&=& {\qh_i^T} \sum_{k=1}^K \qt_k^*  e^{j(\phi_k-\phi_i)} d_i +
%      n_i\nonumber\\
%     &=& {\qh_i^T} \sum_{k=1}^K \qt_1^* e^{j(\phi_1-\phi_i)}  e^{j(\phi_k-\phi_1)} d_i +
%      n_i\nonumber\\
%      &=& {\qh_i^T}e^{j(\phi_1-\phi_i)}    K \qt_1^*   d_i +
%      n_i\nonumber\\
%      &=& {\tilde \qh_i^T}   \qw^*  d_i +
%      n_i,
% \eea
\end{proof}
{ Note that different from the classical multicast
beamforming design, which is non-convex and difficult to solve
\cite{Sidir},} (\ref{eqn:MC}) is a convex problem with a quadratic
objective function and $2K$ linear constraints thus easily solvable.
 In a similar fashion, the  problem \eqref{CIBCrel2} will have a
 similar multicast formulation.

Theorem 1 provides useful insight into the structure of the
precoding vectors. It tells us that the original broadcast problem
now becomes a multicast problem if interference is utilized
constructively. This is understandable because we take into account
the correlation of the originally independent data streams,
therefore the transmission can be re-formulated as sending a common data stream to all users,
and re-shaping the channel and the resulting signal overlap, so that the intended data is delivered to each receiver. More importantly, the multicast
problem contains only a single  vector $\qw$ and is therefore more easily
solved than the original broadcast problem. The rest of this section
is devoted to deriving an efficient algorithm to solve
(\ref{eqn:MC}).

\subsection{Real Representation of the Problem}
For  convenience, we separate the real and imaginary parts of each
complex notation as follows
 \be\label{eqn:real:imag}
    \tilde \qh_i = {\tilde{\qh_R} }_i + j {\tilde{\qh_I} }_i , \qw = \qw_R +j\qw_I,
 \ee
 where ${\tilde{\qh_R} }_i= \mbox{Re}(\tilde \qh_i), {\tilde{\qh_I} }_i= \mbox{Im}(\tilde \qh_i), \qw_R = \mbox{Re}(\qw), \qw_I = \mbox{Im}(\qw), j=\sqrt{-1}$.

 Further define real-valued vectors
 \be\label{eqn:real:vector}
 \qf_i=[\tilde {\qh_R}_i; \tilde{\qh_I}_i], \qw_1 \triangleq [\qw_I; \qw_R], \qw_2
= [\qw_R; -\qw_I]\ee

It is easy to verify that $\qw_1 = \Pi\qw_2$, where $\Pi \triangleq
[\qzero_{N} ~~ \qI_N ; -\qI_N ~~\qzero_N]$.  $\qzero_{N}$ and
$\qI_N$ denote $N\times N$ all-zero matrix and identity matrix,
respectively.
 With the new notations, we can  express the real and imaginary
 parts in (\ref{eqn:constraint}) as follows
 \bea
    \mbox{Re} ({\tilde\qh_i^T} \qw) = \qf_i^T \qw_2, \mbox{Im} ({\tilde\qh_i^T} \qw)   =\qf_i^T \qw_1= \qf_i^T
    \Pi\qw_2.
 \eea

 As a consequence,  the constraint in (\ref{eqn:constraint}) can be rewritten as
 \be
   | \qf_i^T \Pi\qw_2|  \le  \qf_i^T \qw_2\tan\theta
   -\sqrt{\Gamma_i
    N_0}\tan\theta, \forall i. \ee

 Define $ \qg_i\triangleq\Pi^T\qf_i$, and we rewrite  (\ref{eqn:MC}) as
  \bea\label{eqn:CIBC:robust4}
    \min_{\{\qw_2\}} &&  \|\qw_2\|^2  \nonumber\\
    \mbox{s.t.}&&    \qg_i^T  \qw_2   -  \qf_i^T \qw_2\tan\theta  +\sqrt{\Gamma_i
    N_0}\tan\theta\le 0, \forall i, \label{eqn:CIBC:robust4:con1} \\
    &&-\qg_i^T  \qw_2   -  \qf_i^T \qw_2\tan\theta  +\sqrt{\Gamma_i
    N_0}\tan\theta\le 0, \forall i.\label{eqn:CIBC:robust4:con2}
 \eea

\begin{table*}\normalsize
\bea\label{eqn:dual}  {\cal
L}(\qw_2,\bm\mu, \bm \nu)&=& \|\qw_2\|^2+ \sum_{i=1}^K \mu_i
 (\qg_i^T  \qw_2   -  \qf_i^T \qw_2\tan\theta  -\sqrt{\Gamma_i
    N_0}\tan\theta) \\
&&+ \sum_{i=1}^K \nu_i\left( -\qg_i^T  \qw_2   -  \qf_i^T
\qw_2\tan\theta -\sqrt{\Gamma_i
    N_0}\tan\theta\right)\nonumber\\
&=& \|\qw_2\|^2
 + \left(\sum_{i=1}^K  ( \mu_i -\nu_i ) \qg_i^T - \sum_{i=1}^K  (
\mu_i +\nu_i ) \qf_i^T\tan\theta\right)  \qw_2 + \sqrt{\Gamma_i
    N_0}\tan\theta \sum_{i=1}^K (\mu_i+\nu_i)\nonumber\\
    &=& \|\qw_2\|^2
 +  \sum_{i=1}^K \mu_i\left((\qg_i^T-\qf_i^T\tan\theta  )\qw_2+ \sqrt{\Gamma_i
    N_0}\tan\theta \right)  + \sum_{i=1}^K\nu_i \left(-(\qg_i^T+\qf_i^T\tan\theta)\qw_2+ \sqrt{\Gamma_i
    N_0}\tan\theta \right). \nonumber\eea 
    \end{table*}

\subsection{The Dual Problem}
 To further simplify the optimization let us formulate
its dual problem. We let $\bm\mu, \bm\nu \ge\qzero$  be the dual
vector variables associated with the two sets of constraints in
(\ref{eqn:CIBC:robust4:con1}) and (\ref{eqn:CIBC:robust4:con2})
respectively, and consider the Lagrangian of
(\ref{eqn:CIBC:robust4}) as \eqref{eqn:dual}

    The dual objective is
thus given by $\min_{\qw_2}{\cal L}(\qw_2,\bm\mu, \bm \nu)$.
 Setting $\frac{\partial \cal L(\qw_2,\bm\mu, \bm \nu)}{\partial \qw_2}=\qzero$, we obtain
 the structure of the optimal $\qw_2$ below:
  \be\label{eqn:opt:w2}
     \qw_2^*=\frac{\sum_{i=1}^K  ( \mu_i -\nu_i ) \qg_i^T - \sum_{i=1}^K  (
\mu_i +\nu_i) \qf_i^T}{2}. \ee

It is not surprising to see that the optimal $\bar\qw_2$ is the
linear combination of the channel coefficients.

Substituting  (\ref{eqn:opt:w2}) into ${\cal L}(\qw_2,\bm\mu, \bm
\nu)$ leads to 

\bea
{\cal L}(\bm\mu, \bm \nu) &=& -\frac{\|\sum_{i=1}^K
( \mu_i -\nu_i ) \qg_i^T - \sum_{i=1}^K  ( \mu_i +\nu_i)
\qf_i^T\|^2}{4} \nonumber\\
&&+ \sqrt{\Gamma_i
    N_0}\tan\theta \sum_{i=1}^K (\mu_i+\nu_i).\eea

    Define $\bm\lambda=[\bm\mu; \bm\nu]$ and rearrange the terms, and we can obtain the following dual
    problem:
 \bea\label{eqn:MC:dual}
    \max_{\bm\lambda\ge \qzero} && -\frac{\|\qA\bm\lambda \|^2}{4} + \sqrt{\Gamma_i
    N_0} \qone^T \bm\lambda,
 \eea
 where $\qA\triangleq[\qf-\qg ~\qf+\qg]$ is an ${2N\times 2K}$ real matrix, and  the $i$-th column of f and g are defined in (21) and after (23).\color{black}

 The problem (\ref{eqn:MC:dual}) is a non-negative least-squares
 problem. It has wide applications but is known to be a challenging problem \cite{nnls}. Without the non-negative constraint, its solution
 is given by
 \be\label{eqn:lambda}
    \bm\lambda^* = 2\sqrt{\Gamma  N_0}(\qA^T\qA)^{-1}\qone.
 \ee

 Based on the above results, we have the
 following corollary to provide useful insight. % corollary
\begin{corollary}
 If $\bm\lambda^*$ is a positive vector \color{black}, the optimal solution to
 the original problem  (\ref{CIBCrelc3}) with constructive
 interference is achieved when $\mbox{Im} \left({\qh_i^T} \sum_{k=1}^K \qt_k  e^{j(\phi_k-\phi_i)}\right)=0, \forall
 i$. The optimal precoding solution can be obtained from (\ref{eqn:lambda})
 and (\ref{eqn:opt:w2}).
\end{corollary}
This is can be explained by the fact that $\bm\lambda$ is the dual
variable and needs to satisfy the complementary slackness conditions
in (\ref{eqn:CIBC:robust4}): \bea
    &&\mu_i(\qg_i^T  \qw_2   -  \qf_i^T \qw_2\tan\theta
    +\sqrt{\Gamma_i
    N_0}\tan\theta)= 0, \forall i,\\
    &&\nu_i(-\qg_i^T  \qw_2   -  \qf_i^T \qw_2\tan\theta  +\sqrt{\Gamma_i
    N_0}\tan\theta)= 0, \forall i.
 \eea

 When $\mu_i>0, \nu_i>0$, this implies that
 \bea
    &&\qg_i^T  \qw_2   -  \qf_i^T \qw_2\tan\theta  +\sqrt{\Gamma_i
    N_0}\tan\theta =\nonumber\\
    && -\qg_i^T  \qw_2   -  \qf_i^T \qw_2\tan\theta  +\sqrt{\Gamma_i
    N_0}\tan\theta=0,
 \eea
 and it follows that $ \qg_i^T  \qw_2=0$ or $\mbox{Im} \left({\qh_i^T} \sum_{k=1}^K \qt_k  e^{j(\phi_k-\phi_i)}\right)=0, \forall
 i$.

\subsection{A Gradient Projection Algorithm to Solve (\ref{eqn:MC:dual}) }

The general solution to (\ref{eqn:MC:dual})  is difficult to derive.
Here we propose a  gradient projection algorithm to solve it. The
gradient projection algorithm is the natural extension of the
unconstrained steepest descent algorithm to bound constrained
problems \cite{Kelley}. To this end, we first rewrite
(\ref{eqn:MC:dual})  as a standard convex problem:
 \bea\label{eqn:MC:dual2}
    \min_{\bm\lambda\ge \qzero} && f(\bm\lambda)\triangleq\frac{\|\qA\bm\lambda \|^2}{4}
    -\sqrt{\bm\Gamma^T
    N_0}   \bm\lambda,
 \eea
 where $\bm\Gamma=[\Gamma_1,~ \cdots,~ \Gamma_K]^T$.
 It is easy to verify that the gradient of $f(\bm\lambda)$ is given
 by
 \be
  \nabla f(\bm\lambda)  =
 \frac{\qA^T\qA\bm\lambda}{2}-\sqrt{\bm\Gamma^T N_0} \qone.
 \ee
 We then propose the following Algorithm 1 to
 solve (\ref{eqn:MC:dual}).  Once the optimal dual solution $\bm\lambda_{\sf opt}$ is found, the
 original precoding solution can be obtained from  (\ref{eqn:opt:w2}).

 \begin{algorithm}[]
\caption{Efficient Gradient Projection Algorithm to solve
(\ref{eqn:MC:dual2}) }\label{alg:ESB}
\begin{algorithmic}[1]
\STATE {\bf Input:} $ {\bf h},\qd, \Gamma,N_0$

\STATE {\bf begin}

\STATE \hspace*{3mm} Initialize arbitrarily
$\boldsymbol{\lambda}^{(0)}\ge{\bf 0}$.

 \STATE \hspace*{3mm}  In the $n-$th iteration, update $\bm\lambda$:
\be \small \boldsymbol{\lambda}^{(n)} = \max\left( \bm\lambda^{(n-1)} - a_n\left(\frac{\qA^T\qA\bm\lambda^{(n-1)}}{2}-\sqrt{\bm\Gamma^T N_0}\qone \right),\qzero\right), \color{black} \ee
\hspace*{3mm} where the step size $a_n$ can be chosen according to
the Armijo rule or some other line search scheme.

\STATE \hspace*{3mm} Go back to line 4 until convergence.

\STATE {\bf end}

\STATE {\bf Output:} $\bm\lambda_{\sf opt}$.
\end{algorithmic}
\end{algorithm}

\section{Constructive Interference Optimization for SINR Balancing}

The above concept of constructive interference can be applied to the SINR balancing problem of \eqref{eqn:SB2}. The problem for the case of constructive interference can be reformulated as
 \bea\label{eqn:CISB}
    \max_{\{\qt_i, \Gamma_t\}} && \Gamma_t \nonumber\\
    \mbox{s.t.}
    &&\left|\angle \left({\qh_i^T} \sum_{k=1}^K \qt_k  d_k\right)-\angle(d_i)\right|\leq\vartheta, \forall i\nonumber\\
     &&\mbox{Re}\left({\qh_i^T} \sum_{k=1}^K \qt_k  e^{j(\phi_k-\phi_i)}\right)\ge
     \sqrt{\Gamma_tN_0}, \forall i\nonumber\\
    && \left\| \sum_{k=1}^K \qt_k e^{j(\phi_k-\phi_1)}\right\|^2  \leq P
 \eea
{which is not a convex problem because
$\sqrt{\Gamma_t}$ is concave.
 However, this can be simply resolved by replacing it with a new
 variable, i.e., $\Gamma_{t2} = \sqrt{\Gamma_t}$. Then we obtain an
 equivalent problem of \eqref{eqn:CISB2}}
 
 \begin{table*}\normalsize
 \bea\label{eqn:CISB2}
    \max_{\{\qt_i, \Gamma_{t2}\}} && \Gamma_{t2} \nonumber\\
    \mbox{s.t.}
    &&\left|\mbox{Im} \left({\qh_i^T} \sum_{k=1}^K \qt_k  e^{j(\phi_k-\phi_i)}\right)\right|\le\left(\mbox{Re} \left({\qh_i^T}
    \sum_{k=1}^K \qt_k  e^{j(\phi_k-\phi_i)}\right)- \Gamma_{t2} \sqrt{N_0}\right)\tan\theta, \forall i\nonumber\\
    && \left\| \sum_{k=1}^K \qt_k e^{j(\phi_k-\phi_1)}\right\|^2  \leq P
 \eea
 \end{table*}

 It can be seen that, as opposed to
 the conventional SINR balancing problem \eqref{eqn:SB2}, which is non-convex,
 the formulation in \eqref{eqn:CISB2} is convex and can be solved by standard convex optimization techniques.

\subsection{ Multicast Formulation of \eqref{eqn:CISB}}

 In a similar fashion to the power minimization problem, the optimization in \eqref{eqn:CISB} can be re-cast as a multicasting problem according to the following theorem.

\begin{theorem}
 Problem (\ref{eqn:CISB2}) is equivalent to the multicast problem
 below:
 \bea\label{eqn:MC3}
    \max_{\{\qw, \Gamma_{t2}\}} && \Gamma_{t2} \nonumber\\
%     \mbox{s.t.} &&  - \left(\mbox{Re} \left({\tilde\qh_i^T} \qw  \right) - \sqrt{\Gamma_iN_0} \right)\tan \theta  \le  \nonumber\\
&&\left|\mbox{Im} \left({\tilde\qh_i^T}   \qw \right) \right| \le \left(\mbox{Re} \left({\tilde\qh_i^T} \qw  \right) - \Gamma_{t2}\sqrt{N_0} \right)\tan \theta, \forall i,   \nonumber\\
    &&\|\qw\|^2 \leq P
\eea

where $\tilde \qh_i = \qh_i e^{j(\phi_1-\phi_i)}$. To be more
specific, if the optimal solutions to (\ref{eqn:CISB}) and
(\ref{eqn:MC}) are $\{\qt^*_i\}$ and $\qw^*$, respectively, then
they have the following relation:
\bea
  \qt_1^* &= &\frac{\qw^*}{K}, \\
  \qt_k &=& \qt_1^*e^{j(\phi_1-\phi_k)} =\frac{\qw^*e^{j(\phi_1-\phi_k)} }{K}, k = 2,
  \cdots, K.
\eea
\end{theorem}
\begin{proof} The proof follows the one for Theorem \ref{T1}
\end{proof}
 The above is a standard SOCP problem which can be efficiently solved using known approaches.

 \section{Robust Power Minimization with bounded CSI errors }

\subsection{Channel Error Model and Problem Formulation}
 In this section, we study the robust procoding design
when the CSI is imperfectly known. We model user $i$'s actual
channel as
 \be
     \qh_i = \hat\qh_i +  \qe_i, \forall k,
 \ee
 where   $\hat\qh_i$ denotes the CSI estimates known to the BS, and
 $\qe_i$ represents the CSI uncertainty within the spherical set $\mathcal{U}_i=\{\qe_i: \|\qe_i\|^2\le \delta_i^2\}.$

 We assume that the BS has no knowledge about $\qe_i$ except for its error bound $\delta_i^2$ thus we take  a worst-case approach for the transmit precoding design to guarantee the resulting
 solution is  robust to all possible channel uncertainties within $\mathcal{U}_i$. The specific robust design problem is to minimize the overall transmit
power $P_T$ for ensuring the users' individual SINR constraints by
optimizing the precoding vector $\{\qt_i\}$,  i.e.,
\begin{equation}\label{prob:miso}
\begin{aligned}
\min_{\qw} &~P_T ~~ \mbox{s.t.} &~\mbox{SINR}_i\ge\Gamma_i~, \forall
\qe_i\in\mathcal{U}_i, \forall i.
\end{aligned}
\end{equation}

\subsection{Conventional Robust Precoding }
 In the conventional multiuser MISO systems, the total transmit
 power is $P_T=\sum_{i=1}^K\|{\bf t}_i\|^2$ and robust design
 problem can be expressed as
\bea \label{prob:miso2}
\min_{\{\qt_i\}} &&\sum_{i=1}^K\|{\bf t}_i\|^2 \nonumber \\
 \mbox{s.t.}&&
\frac{\left|{
\qh}^T_i{\qt}_i\right|^2}{\displaystyle\sum_{k=1\atop k\ne
i}^{K}\left|  \qh_i^T{\bf t}_k\right|^2+N_0} \ge\Gamma_i, \forall
\qe_i\in\mathcal{U}_i, \forall k.
\eea

The robust precoding design is characterized by the following
theorem \cite{GanRobust}.
\begin{theorem}\label{worst2sdp}
The  robust beamforming problem (\ref{prob:miso2}) can be relaxed to
the following semi-definite programming (SDP) problem
\begin{equation}\label{eqn:sdp:miso}
\begin{aligned}
&~~\min_{\left\{{\bf T}_i\succeq\qzero, s_i\ge 0\right\}}  \sum_{k=1}^K{\rm trace}(\qT_i)\\
&~~ \mbox{\rm s.t.} \left[\begin{array}{cc}
\hat{\bf h}_i^*\qQ_i\hat{\bf h}_i^T-\gamma_iN_0-s_i \delta_i^2  & \hat{\bf h}_i^*\qQ_i\\
\qQ_i\hat{\bf h}_i^T & \qQ_i+\delta_i^2 \qI \\
\end{array}\right]
\succeq \qzero\quad\forall k,
\end{aligned}
\end{equation}
where
\begin{equation*}
\qQ_i\triangleq\qT_i-\Gamma_i\sum_{k=1\atop k\ne i}^M\qT_n~~\forall
k.
\end{equation*}
The problem (\ref{eqn:sdp:miso}) is convex and hence can be
optimally solved. The resulting objective value of
(\ref{eqn:sdp:miso}) provides a lower bound for the conventional
power minimization.
\end{theorem}

{\it Remark: When the SDP relaxation is tight, or
(\ref{eqn:sdp:miso}) returns all rank-1 solution $\{\qT_i\}$, then
the optimal solution $\{\qt_i\}$ to solve (\ref{prob:miso2}) can be
obtained by matrix decomposition such that $\qT_i =\qt_i\qt_i^\dag,
\forall i$; otherwise, the required power in (\ref{prob:miso2}) is
always higher than that in (\ref{eqn:sdp:miso}).
%We will nevertheless use objective value of (\ref{eqn:sdp:miso}) as the
%power required by the conventional scheme for comparison.

}

\subsection{Robust Precoding based on Constructive Interference}

Based on the multicast formulation of the power minimization problem(\ref{eqn:MC}), the worst-case robust design for the case
of constructive interference becomes
  \bea\label{eqn:CIBC:robust}
    \min_{\qw} &&  \|\qw\|^2 \nonumber\\
    \mbox{s.t.}
    && \left|\mbox{Im}  \left({\tilde\qh_i^T} \qw \right)\right|- \left( \mbox{Re} \left({\tilde\qh_i^T} \qw\right)  -   \sqrt{\Gamma_i N_0}\right)\tan
    \theta   \le 0, \notag\\
    &&\forall  \|\qe_i \|^2\le\delta_i^2, \forall i.\label{const:angle}
 \eea

 The constraint in (\ref{eqn:CIBC:robust}) is intractable due to the infinite number of error vectors.
  Below we show how to tackle it using convex optimization
  techniques. For ease of composition, we omit the user index.

    First notice that  robust constraint  (\ref{const:angle}) can be guaranteed by the modified constraint below:
  \be
    \max_{\|\qe\|^2\le
    \delta^2}\left( \left|\mbox{Im} \left({\tilde\qh^T} \qw \right)\right|- \left( \mbox{Re} \left({\tilde\qh^T} \qw\right)  -   \sqrt{\Gamma N_0}\right)\tan
    \theta\right) \le 0. \label{const:angle:no:index}
  \ee
   We again separate the real and imaginary parts of each complex notation as follows
 \bea\label{eqn:real:imag2}
    \tilde \qh &=& \tilde\qh_R + j\tilde\qh_I\nonumber\\
    &=& \hat \qh_R + j\hat\qh_I + \qe_R +
    j\qe_I.
 \eea

 Further define real-valued error vector and channel estimation
 vector
 \be\label{eqn:real:vector2}\bar\qe \triangleq [\qe_R; \qe_T],   \hat\qf = [\hat\qh_R;
 \hat\qh_I].\ee
 Apparently $\|\bar\qe\|^2\le  \delta^2$.
 With the new notations, we can re-express the real and imaginary
 parts in (\ref{const:angle:no:index}) as follows
 \bea
    \mbox{Im} ({\tilde\qh^T} \qw) &=&\hat \qh_R^T \qw_I + \hat \qh_I \qw_R +
    \qe_R^T \qw_I + \qe_I^T\qw_R,\nonumber\\
    &= &\hat\qf^T \qw_1  + \bar\qe^T \qw_1;\\
    \mbox{Re} ({\tilde\qh^T} \qw)& =&\hat \qh_R^T \qw_R - \hat \qh_I \qw_I +
    \qe_R^T \qw_R -\qe_I^T\qw_I,\nonumber\\
    &= &\hat\qf^T \qw_2  + \bar\qe^T \qw_2.
 \eea

 As a consequence,  (\ref{const:angle:no:index}) can be rewritten as
 \bea\label{eqn:robust:angle:1}
   \max_{\|\bar\qe\|^2\le  \delta^2}  |\hat\qh^T \qw_1  + \bar\qe^T \qw_1| - (\hat\qh^T \qw_2  + \bar\qe^T
    \qw_2)\tan\theta \nonumber\\ 
    \quad \quad \quad \quad \quad + \sqrt{\Gamma N_0}\tan\theta \le 0.
 \eea
 The above constraint is equivalent to the following two
 constraints:
 \bea\label{eqn:robust:angle:2}
   \max_{\|\bar\qe\|^2\le  \delta^2}    \hat\qh^T \qw_1  + \bar\qe^T \qw_1  - (\hat\qh^T \qw_2  + \bar\qe^T
    \qw_2)\tan\theta \nonumber\\ 
    \quad \quad \quad \quad \quad + \sqrt{\Gamma N_0}\tan\theta \le 0,\\
       \max_{\|\bar\qe\|^2\le  \delta^2}    -\hat\qh^T \qw_1  - \bar\qe^T \qw_1  - (\hat\qh^T \qw_2  + \bar\qe^T
    \qw_2)\tan\theta \nonumber\\ 
    \quad \quad \quad \quad \quad+ \sqrt{\Gamma N_0}\tan\theta \le 0,
 \eea
  whose robust formulations are given by
 \bea
   \hat\qh^T \qw_1 -  \hat\qh^T \qw_2\tan\theta  +   \delta\|\qw_1-
    \qw_2\tan\theta\| \nonumber\\ 
    \quad \quad \quad \quad \quad + \sqrt{\Gamma N_0}\tan\theta \le 0,\label{q1}\\
            -\hat\qh^T \qw_1  -  \hat\qh^T \qw_2\tan\theta
          + \delta\|\qw_1+ \qw_2\tan\theta\| \nonumber\\ 
    \quad \quad \quad \quad \quad + \sqrt{\Gamma N_0}\tan\theta \le
          0.\label{q2}
 \eea

\iffalse
 The left-hand-side of the constraint (\ref{eqn:robust:angle:1}) is
 difficult to evaluate because the maximization problem with regard
 to $\bar\qe$ is non-convex. To circumvent this, we propose to use  instead the following sufficient condition to guarantee
 (\ref{eqn:robust:angle:1}):
  \be
   |\hat\qh^T \qw_1| + \delta |\qw_1| - (\hat\qh^T \qw_2 - \delta|\qw_2|)\tan\theta+ \sqrt{\Gamma N_0}\tan\theta\le
   0,
 \ee
  \be\label{eqn:robust:angle:2}
   |\hat\qh^T \qw_1| + \delta |\qw_1| + \delta|\qw_2|\tan\theta \le \hat\qh^T \qw_2\tan\theta  -\sqrt{\Gamma N_0}\tan\theta,
 \ee
 which is a SOCP constraint.

  We also observe that the robust constraint (\ref{eqn:robust:angle:2}) has an additional advantage of
 ensuring that
 \be  {\tilde\qh^T} \qw  -   \sqrt{\Gamma N_0}  =
 (\hat\qh^T \qw_2  + \bar\qe^T \qw_2)-\sqrt{\Gamma N_0} \ge 0, \forall \|\bar\qe\|^2\le
    \delta^2.\ee
 \fi

 Finally we reach the following robust  problem formulation
  \bea\label{eqn:CIBC:robust2}
    \min_{ \qw_1, \qw_2 } &&  \|\qw_1\|^2 \nonumber\\
    \mbox{s.t.}
     && \mbox{Constraints (\ref{q1}) and (\ref{q2})}, \forall i,\nonumber\\
     %&&|\hat\qh_i^T \qw_1| + \delta_i |\qw_1| + \delta_i|\qw_2|\tan\theta
      %\notag \\
     %&&\le \hat\qh_i^T \qw_2\tan\theta  -\sqrt{\Gamma
    %N_0}\tan\theta, \forall i,\nonumber\\
    &&    \qw_1 = \Pi\qw_2.
 \eea
 Note that (\ref{eqn:CIBC:robust2}) is standard SOCP problem
 therefore can be efficiently solved. After we obtain the optimal $\qw_1^*,
 \qw_2^*$, the robust solution $\qw^*$ can be determined using the
 relation in (\ref{eqn:real:vector}).

\subsection{Robust SINR Balancing}

For completeness, we also study the robust SINR balancing given
total BS power constraint. Similar procedures can be taken as the
above to derive it, therefore we give the problem formulation
directly below
  \bea\label{eqn:CIBC:SINR:robust}
    \min_{ \qw_1, \qw_2,\Gamma_{t2}} &&  \Gamma_{t2} \nonumber\\
    \mbox{s.t.}&&
    % |\hat\qh_i^T \qw_1| + \delta_i |\qw_1| + \delta_i|\qw_2|\tan\theta \notag \\
    % && \le \hat\qh_i^T \qw_2\tan\theta  - \Gamma_{t2}\sqrt{
    %N_0}\tan\theta,  \forall i,\nonumber\\
    \mbox{Constraints (\ref{q1}) and (\ref{q2})}, \forall i,\nonumber\\
    && \qw_1 = \Pi\qw_2,\nonumber\\
    && \|\qw_1\|^2\le P_T.
 \eea
The above is a typical SOCP problem that can be solved using standard approaches. Suppose the optimal $\Gamma_{t2}$ is $\Gamma^*_{t2}$, then the
maximum SINR value becomes ${\Gamma^*_{t2}}^2$.

\color{black}
\section{Numerical Results}

This section presents numerical results based on Monte Carlo
simulations of the proposed optimization techniques,  termed as CI
in the following, and conventional precoding based on optimization
for the frequency flat Rayleigh fading statistically uncorrelated
multiuser MISO channel for both cases of perfect and erroneous CSI. Systems
with BPSK, QPSK and 8PSK modulation are considered while it is clear
that the benefits of the proposed technique extend to higher order
modulation. To focus on the proposed concept, we compare the
proposed CI optimization solely to the conventional optimization  of
\cite{Ott},\cite{Boche}, while it is clear the the benefits of the
proposed concept extend to modified optimization designs in the
literature, by direct application of the constructive interference
concept and the adaptation of the relevant QoS constraints. {We use `$N\times K$'  to denote a mutiuser
MISO system with $N$ transmit antennas and $K$ single-antenna
terminals. Unless otherwise specified, QPSK  is the default
modulation scheme.}

\begin{figure}
    \centering
        \includegraphics[width=\fsize\textwidth]{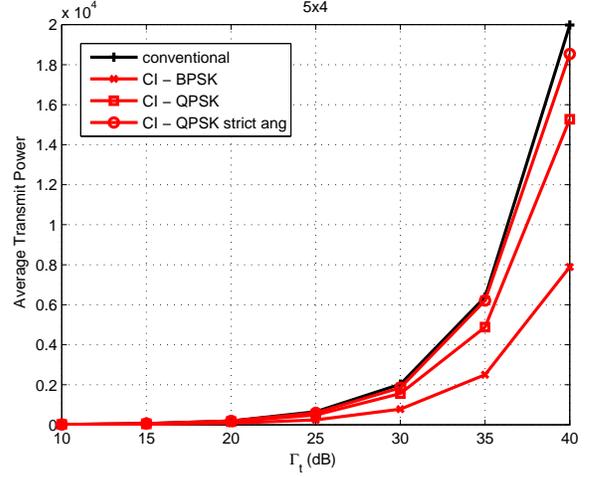}
    \caption{{Transmit power vs. $\Gamma_t$ for conventional and CI, $K$=4, $N$=5}}
    \label{5x4}
\end{figure}
 \begin{figure}
    \centering
        \includegraphics[width=\fsize\textwidth]{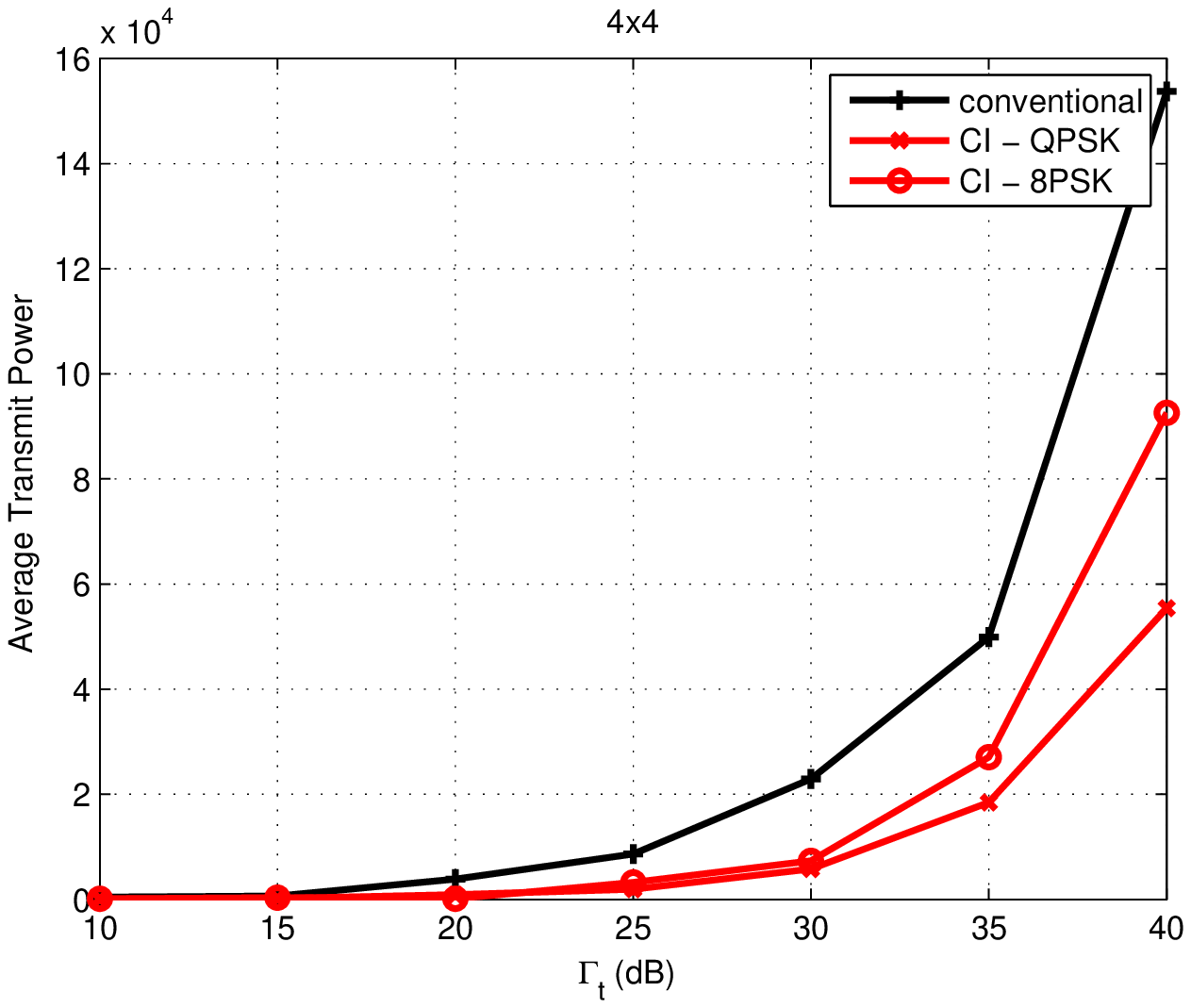}
    \caption{{Transmit power vs. $\Gamma_t$ for conventional and CI, $K$=4, $N$=4}}
    \label{4x4}
\end{figure}
 \begin{figure}
    \centering
        \includegraphics[width=\fsize\textwidth]{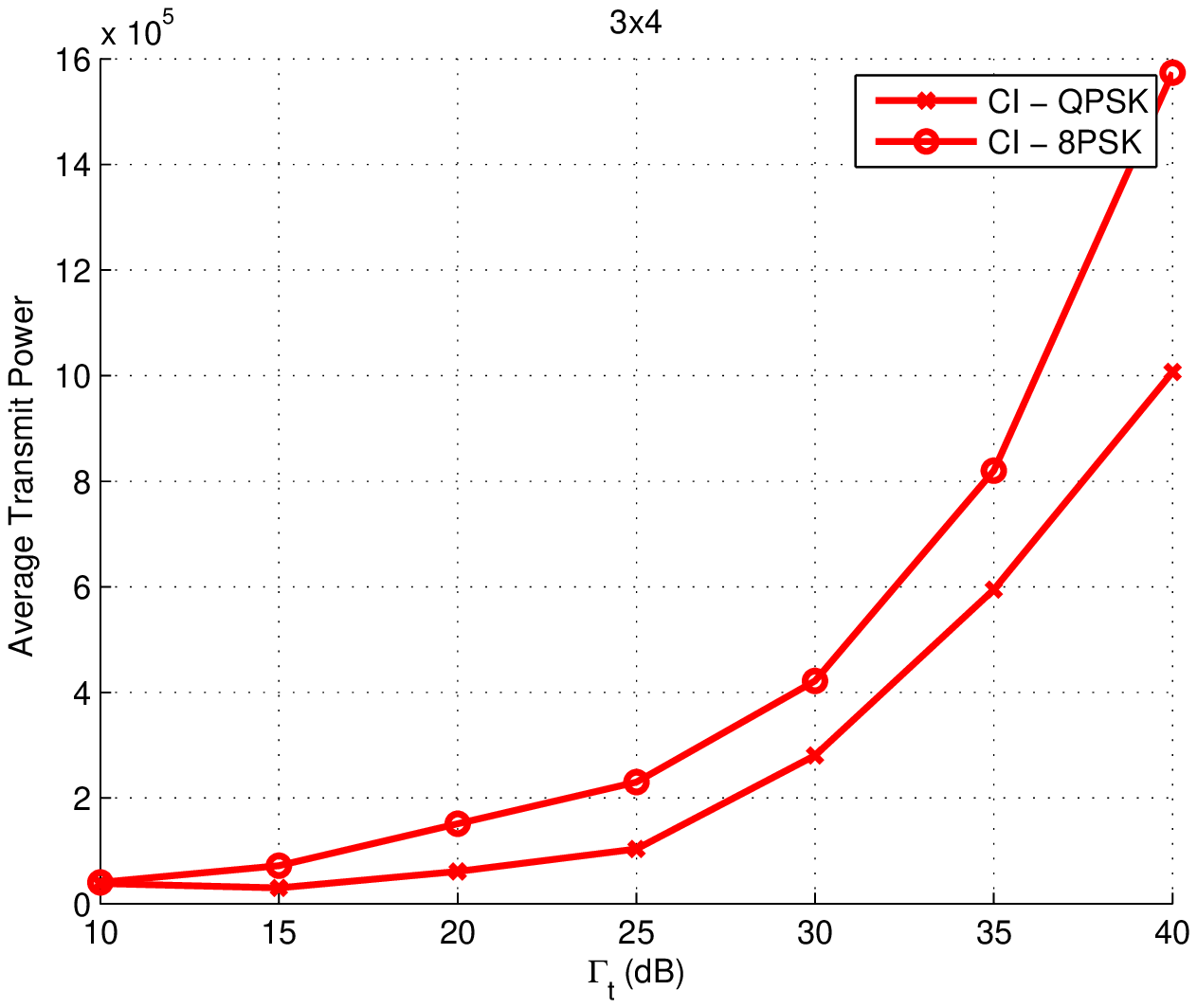}
    \caption{{Transmit power vs. $\Gamma_t$ for conventional and CI, $K$=4, $N$=3}}
    \label{3x4}
\end{figure}

 \begin{figure}
    \centering
        \includegraphics[width=\fsize\textwidth]{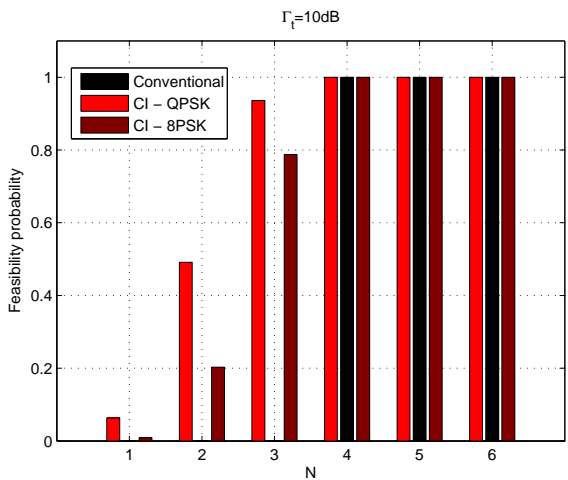}
    \caption{{Feasibility probability vs. $N$ for conventional and CI, $K$=4, $\Gamma_t$=10dB}}
    \label{nx4feas}
\end{figure}
First, in Fig. \ref{5x4} we compare the average transmit power with
the conventional optimization and the proposed optimization problems
\eqref{CIBC2} and \eqref{CIBC3} in the $5\times4$ scenario for BPSK and QPSK, respectively. Power savings
of up to 50\% can be seen. It can also be observed that the relaxed
optimization \eqref{CIBCrelc3} results in significant power gains
compared to the strict angle constraints in \eqref{CIBCstrict2}. The
same trend can be observed for the $4\times4$ systems case of Fig. \ref{4x4} where the conventional
optimization results in a solution with increased transmit power,
{because less transmit antennas are available at the
BS}. The gains in this case are amplified for the
proposed relaxed optimization. The transmit power is also shown in
Fig. \ref{3x4} for the $3\times4$ scenario where the conventional
precoder is inapplicable, as the optimization in \eqref{eqn:BC} is
infeasible for $N<K$. Remarkably, the proposed optimization problem is
feasible in the 92.6\% of the cases. The relaxed nature of the
problem indeed leads to larger feasibility regions. To illustrate
the extended feasibility region for the proposed optimization
problems, Fig. \ref{nx4feas} shows the feasibility probability of a $K=4$ user system with
respect to the number of transmit antennas. It can be seen that
while the conventional optimization is only feasible for $N\geq K$,
the proposed can be feasible for lower values of $N$.

 The complexity of the proposed power minimization problem is addressed in Fig. \ref{5xKtime} for the system with $N=5$. While analytical complexity expressions are hard to derive for optimization-based precoding, here the complexity of the solvers with the broadcast (BC) formulation of \eqref{CIBCrelc3}, multicast (MC) formulation of \eqref{eqn:MC} and gradient projection solver in Algorithm 1 is shown in terms of the average execution time of the optimization algorithm, with increasing numbers of users. While the BC and MC show a similar computational complexity, the gradient projection solver in Algorithm 1 offers a significant reduction in the involved complexity down to less than $15\%$, which further motivates the multicast formulation of the optimization problem. Still, we note that the proposed optimization needs to be performed on a symbol-by-symbol basis. Accordingly, the proposed may involve excess complexity compared to conventional precoding optimization especially for slow fading scenarios, where the convectional channel-dependent precoding may not require the optimization to be performed frequently. Accordingly, for the slow fading scenario where the channel is constant over a transmission frame and for the example of an LTE frame with 14 symbol time slots, the proposed precoding with the gradient projection solver translates to a doubling of complexity per frame ($14 \times 15\% = 210\%$) for the proposed scheme w.r.t. conventional precoding optimization. However, for the obtained transmit-power benefits as shown in our results, the end complexity increase is definitely worthwhile the performance benefits offered. In fact, in terms of the ultimate metric of power efficiency of the transmitter, we note that for an LTE base station the transmit power is measured in the order of around 20Watts, while the power consumption of the DSP processing is typically of the order of hundreds of milliWatts. Since with the proposed precoding our results show a halving of the transmit power at roughly double the DSP power, the gains in the power efficiency for the proposed scheme w.r.t. conventional precoding are therefore significant.\color{black}

 \begin{figure}[tb]
    \centering
        \includegraphics[width=\fsize\textwidth]{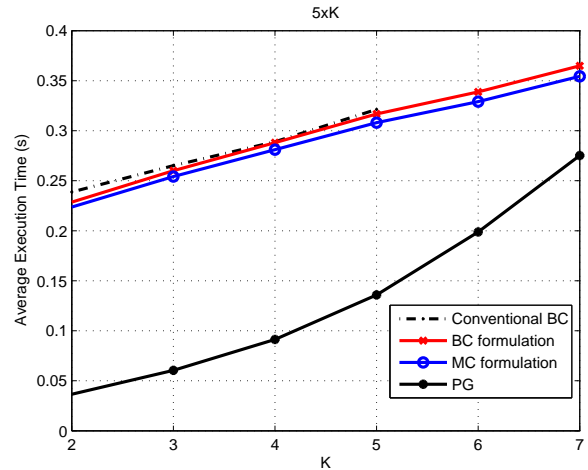}
    \caption{{Average execution time vs. $K$ for CI optimization with the BC formulation of \eqref{CIBCrelc3}, MC formulation of \eqref{eqn:MC} and gradient projection solver in Algorithm 1, $N$=5, QPSK}}
    \label{5xKtime}
\end{figure}
\begin{figure}[tb]
    \centering
        \includegraphics[width=\fsize\textwidth]{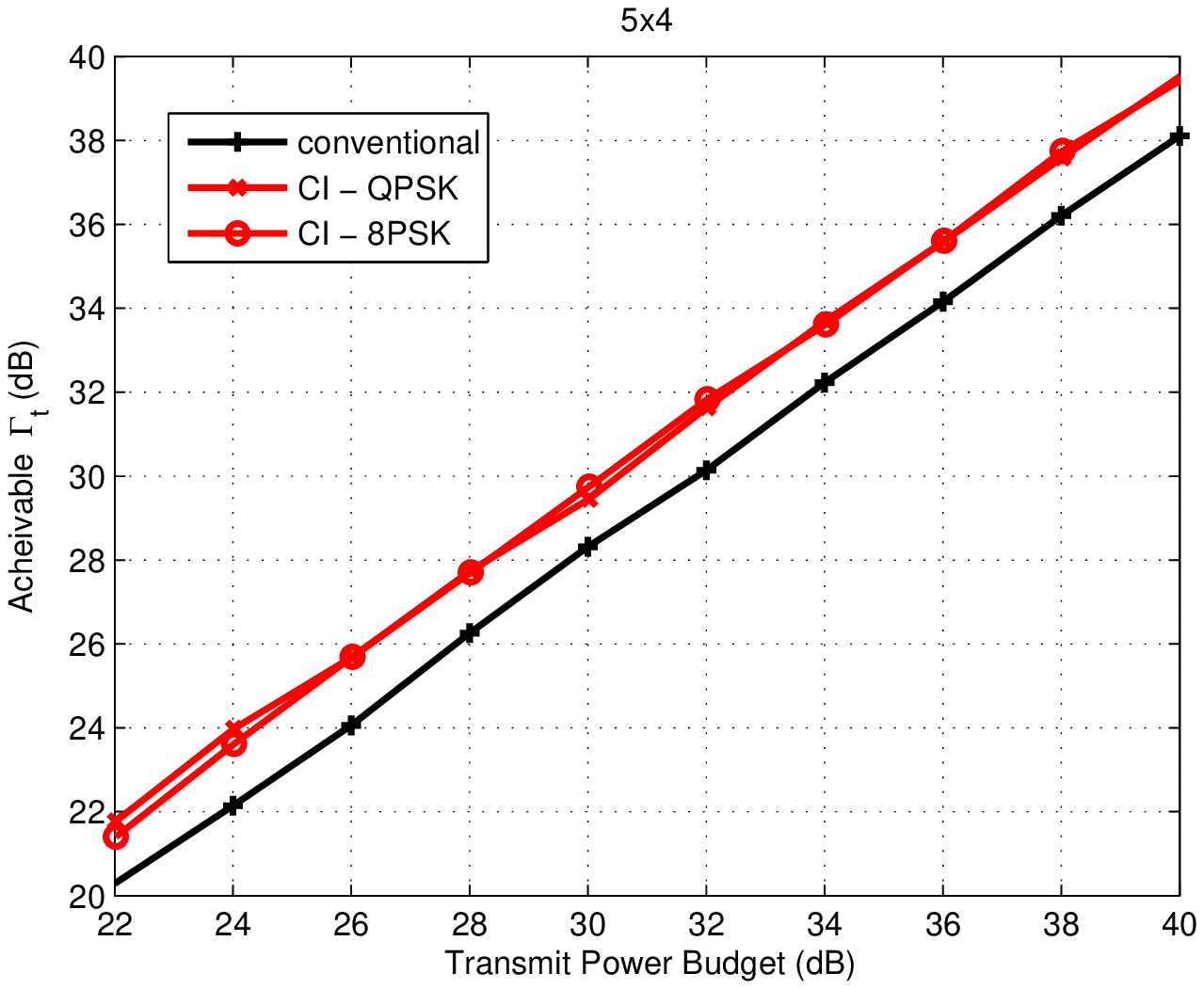}
    \caption{{Achievable $\Gamma_t$ vs. Transmit power budget for conventional and CI, $K$=4, $N$=5}}
    \label{5x4G}
\end{figure}
 \begin{figure}[tb]
    \centering
        \includegraphics[width=\fsize\textwidth]{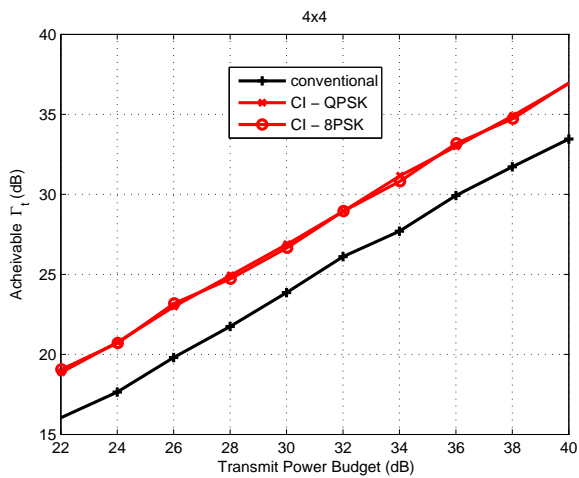}
    \caption{{Achievable $\Gamma_t$ vs. Transmit power budget for conventional and CI, $K$=4, $N$=4}}
    \label{4x4G}
\end{figure}

Figs. \ref{5x4G} and \ref{4x4G} compare the achievable SNR for the SINR balancing
problems \eqref{eqn:SB2} and \eqref{eqn:CISB}. In Fig. \ref{5x4G}
the achievable $\Gamma_t$ is shown for the $5\times4$ MISO where an
SNR gain of about 2dB can be observed. The same trend is observed in
Fig. \ref{4x4G} for the $5\times4$ MISO with an
SNR gain of about 3dB. It is worth noting that for
the case of constructive interference these SNR gains are due to the effect of constructive
interference.

Finally Figs. \ref{4x4CSIdelta} and \ref{4x4CSI} compare the
performance of the proposed CSI-robust CI precoder with the
conventional CSI-robust precoder of \cite{GanRobust} for the
$4\times4$ MISO. Fig. \ref{4x4CSIdelta} shows the obtained transmit
power with increasing CSI error bounds $\delta$ where it can be seen
that for values in the region of $\delta^2=10^{-3}$ the transmit
power increases significantly. On the contrary, the proposed
optimization shows a modest increase in transmit power for high
values of $\delta$ thanks to the relaxed optimization obtained by
exploiting constructive interference. This is also shown in Fig.
\ref{4x4CSI} where the average transmit power is shown with
increasing SNR thresholds, for $\delta^2=10^{-4}$. Again the
proposed shows a constant loss of less than 1dB compared to the
perfect CSI case, while conventional precoding shows an increased
sensitivity to the CSI errors.

\begin{figure}[tb]
\centering
\includegraphics[width=\fsb\textwidth]{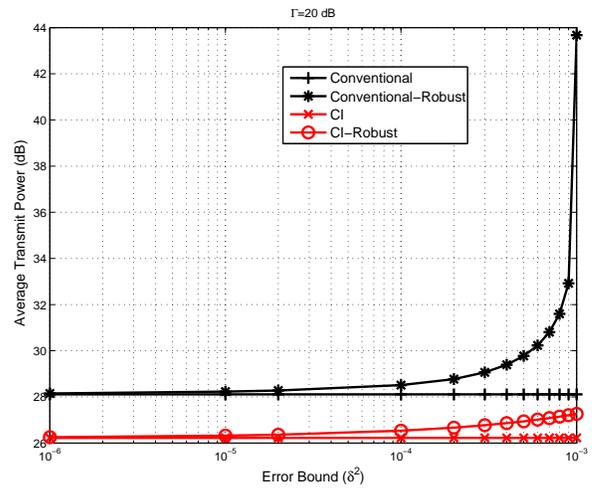}
\caption{Transmit power vs. error bound for conventional and CI, $K$=4, $N$=4,  $\Gamma=20$ dB. \color{black}}
    \label{4x4CSIdelta}
\end{figure}

\begin{figure}[tb]
\centering
\includegraphics[width=\fsb\textwidth]{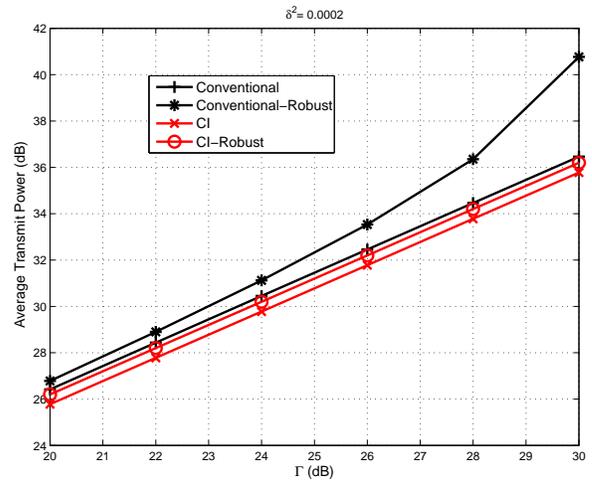}
\caption{Transmit power vs. SINR for conventional and CI, $K$=4, $N$=4, $\delta^2=10^{-4}$. \color{black}}
    \label{4x4CSI}
\end{figure}

\color{black}
\section{Conclusion}
{A number of improved optimization-based precoding designs have been proposed for the multi-user MISO downlink channel. }By taking advantage
of the constructive interference, the proposed schemes deliver
reduced transmit power for given QoS constraints, or equivalently
increased minimum SNRs for a given transmit power budget. Both optimizations were adapted to virtual multicast formulations which were shown to be more efficiently solvable. The
proposed concept was further extended to robust designs for
imperfect CSI with bounded CSI errors. In all cases the proposed
schemes provide superior performance over the
conventional precoding schemes, which confirms the  benefit of constructive interference when the precoders are fully optimized. 

The concept of constructive interference opens up new
opportunities for future work in the optimization of precoding designs. Firstly, the
extension to QAM constellation is non-trivial, and it needs
a careful redesign of constructive interference sectors for determining the relevant optimization constraints.
 Secondly, the robust design in Section VI takes a conservative
approach while the optimal robust precoder design is an important
but challenging problem to be studied. Thirdly, as this work considers a single-cell system, and as
multi-cell cooperation is made practical in systems such as
cloud-RAN \cite{cran}, it is worth investigating the potential of
the proposed scheme in multi-cell environments. Another promising
application scenario is multi-beam satellite systems where both CSI
and data for different beams are available at the gateway station
where forward link beamforming is designed \cite{multibeam}.

\begin{biography}
[{\includegraphics[width=1in,height=1.17in,clip]{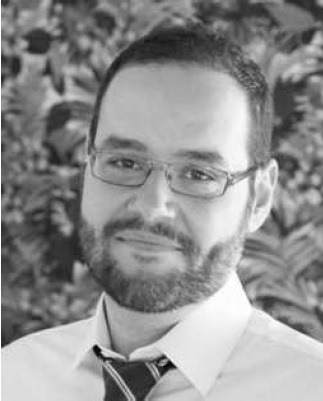}}]
% [{\includegraphics[width=1in,height=1.17in,clip]{CMas2.jpg}}]
{Christos Masouros} %\begin{biography}

(M'06, SM'14), is currently a Lecturer in the Dept. of Electrical \& Electronic Eng., University College London. 
He received his Diploma in Electrical \& Computer Engineering from the University of Patras, Greece, in 2004, MSc by research and PhD in Electrical \& Electronic Engineering from the University of Manchester, UK in 2006 and 2009 respectively. He has previously held a Research Associate position in University of Manchester, UK and a Research Fellow position in Queen's University Belfast, UK. 

He holds a Royal Academy of Engineering Research Fellowship 2011-2016 and is the principal investigator of the EPSRC project EP/M014150/1 on large scale antenna systems. His research interests lie in the field of wireless communications and signal processing with particular focus on Green Communications, Large Scale Antenna Systems, Cognitive Radio, interference mitigation techniques for MIMO and multicarrier communications.\end{biography}

\begin{biography}
[{\includegraphics[width=1in,height=1.3in,clip]{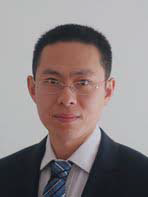}}]
{Gan Zheng} %\begin{biography}
(S'05-M'09-SM'12) 
is currently a Lecturer in School of Computer Science
and Electronic Engineering, University of Essex, UK. 
He received the
 B. E. and the M. E. from Tianjin University, Tianjin, China, in 2002 and 2004,
respectively, both in Electronic and Information Engineering,and the PhD 
degree in Electrical 
and Electronic Engineering from The University 
of Hong Kong, Hong Kong, in 2008. Before he joined University of Essex, he worked as a Research 
Associate
at University College London, UK, and University of Luxembourg,
Luxembourg. His research interests 
include cooperative
 communications, cognitive radio, physical-layer security, full-duplex radio and energy harvesting.  He is the 
first recipient for the 2013 IEEE
Signal Processing Letters Best Paper Award.
\end{biography}

\end{document}